\documentclass[11pt]{article}
\usepackage[utf8]{inputenc}
\usepackage[backend=biber, style=alphabetic, maxbibnames=99]{biblatex}
\usepackage[margin=1in]{geometry}
\usepackage{amsmath, amsthm, amsfonts, amssymb, braket, complexity, comment}
\usepackage{hyperref}
\usepackage{cleveref}
\usepackage{authblk}

\hypersetup{
    colorlinks=true,
    citecolor=cyan,
    linkcolor=magenta,
    }

\allowdisplaybreaks

\newcommand{\defeq}{\stackrel{\mathrm{\scriptscriptstyle def}}{=}}

\DeclareMathOperator{\unit}{unit}
\newcommand{\maxcut}{MaxCut}

\theoremstyle{definition}

\newtheorem{lemma}{Lemma}
\newtheorem{theorem}{Theorem}
\newtheorem{corollary}{Corollary}
\newtheorem{algorithm}{Algorithm}
\newtheorem{problem}{Problem}
\newtheorem{fact}{Fact}
\newtheorem{remark}{Remark}

\begin{document}
\title{No Quantum Advantage in Decoded Quantum Interferometry for MaxCut}
\author{Ojas Parekh\thanks{\href{mailto:odparek@sandia.gov}{odparek@sandia.gov}}}
\affil{Quantum Algorithms and Applications Collaboratory,\\Sandia National Laboratories,\\Albuquerque, NM, USA}
\date{}
\maketitle

\begin{abstract}
Decoded Quantum Interferometry (DQI) is a framework for approximating special kinds of discrete optimization problems that relies on problem structure in a way that sets it apart from other classical or quantum approaches. We show that the instances of MaxCut on which DQI attains a nontrivial asymptotic approximation guarantee are solvable exactly in classical polynomial time. We include a streamlined exposition of DQI tailored for MaxCut that relies on elementary graph theory instead of coding theory to motivate and explain the algorithm.
\end{abstract}

\section{Introduction}
\maxcut{} is a canonical constraint satisfaction problem that has sparked diverse avenues of research. \maxcut{} has also become a testbed problem for quantum-algorithmic approaches for solving or approximating discrete optimization problems, such as quantum annealing and the Quantum Approximate Optimization Algorithm (QAOA)~\cite{FarhiEtAl2014QAOA}. Although approximation guarantees can be derived for the QAOA on \maxcut{}, rigorous approximation advantages remain elusive despite over a decade of extensive study.
A quantum advantage for approximating general instances of \maxcut{} seems unlikely, since we can obtain a $0.878...$-approximation in classical polynomial time~\cite{GoemansWilliamson1995MaxCut}, and a $(0.878...+\varepsilon)$-approximation is \NP-hard under the Unique Games Conjecture~\cite{KhotEtAl2007MaxCut}. This suggests that it is hard for even efficient quantum algorithms to improve upon this classical threshold. 

For ``natural'' discrete optimization problems, why would we \emph{not} expect that there is an $\alpha$ such that an $\alpha$-approximation is possible in classical polynomial time and an $\alpha+\epsilon$-approximation is \NP-hard? What kind of problems would admit a gap between classical approximability and \NP-hardness of approximability, for a quantum advantage to possibly sneak in? In other words, why would a ``natural'' classical discrete optimization problem need a quantum algorithm to complete the story of its approximability?

A recent example of an unexpected provable exponential quantum approximation advantage is for the directed version of \maxcut{} \cite{KallaugherEtAl2025QuantumStreaming,KallaugherEtAl2024StreamingMaxCut}, which is arguably a natural problem. However, in this case the quantum advantage is enabled by considering space advantages in the streaming model instead of more traditional speedups. More broadly, perhaps certain kinds of restricted instances or regimes of natural discrete optimization problems are amenable to quantum approximation advantages. 

Decoded Quantum Interferometry (DQI) potentially presents such an opportunity, since it relies heavily on special problem structure in a way that sets it apart from other classical or quantum approaches~\cite{JordanEtAl2024OptimizationDQI}. DQI is designed to work on constraint satisfaction problems for which an associated decoding problem is efficiently classically solvable. Moreover, it is based on Quantum Fourier Sampling~\cite{FeffermanUmans2016QFS}, which does enable provable or conjectured quantum advantages in many settings. DQI is an exciting development, and there are some examples of problems for which DQI gives a better approximation guarantee than the best classical algorithm currently known. However, the demands that DQI puts on problem structure results in problem regimes unusual and understudied for classical algorithms. The most prominent candidate for a DQI advantage is for Optimal Polynomial Intersection, and there is essentially one nontrivial classical approximation algorithm, that is decades old, currently known for it.

\paragraph{Contributions.} Understanding the power and limitations of DQI is thus an important challenge. We make some progress in establishing limitations for DQI, for the fundamental \maxcut{} problem. \maxcut{} is perhaps the simplest interesting special case of the Max-LINSAT problem considered by Jordan et al.~\cite{JordanEtAl2024OptimizationDQI}, allowing for streamlined presentation and analysis of DQI specialized to \maxcut{}. 

We show that DQI for \maxcut{} exhibits perhaps surprising behavior: on instance families where DQI gets \emph{any} kind of nontrivial approximation guarantee, it is not just that the approximation guarantee can be beaten classically---such instances can be solved \emph{exactly} in classical polynomial time. To derive this result, we concretely demonstrate how the structure necessary for DQI can also be exploited classically. Our observations and classical algorithms are relatively simple, and \maxcut{} may be a useful guidepost for future DQI work. 

DQI is generally not guaranteed to run in polynomial time and requires both feasibility of and an efficient algorithm for a related classical decoding problem. This limits both the problems and types of instances amenable to DQI. In contrast, DQI runs efficiently on any instance of \maxcut{}, since the decoding problem reduces to a perfect matching problem.

The key graph parameter for analysis of DQI for \maxcut{} is the \emph{girth} of the unweighted input graph, which is the length of a shortest cycle. We assume the input comes from an infinite family of graphs with $g$, $n$, and $m$ corresponding to girth, number of vertices, and  number of edges. Based on the more general analysis of DQI in~\cite{JordanEtAl2024OptimizationDQI}, we first observe:
\newtheorem*{theorem:main1}{\Cref{thm:DQI-needs-high-girth}}
\begin{theorem:main1} If the expected number of edges cut in a graph by DQI is $(\frac{1}{2}+\varepsilon)m$, for some constant $\varepsilon > 0$, the girth $g$ must be linear in $m$.
\end{theorem:main1}
The trivial random assignment algorithm for \maxcut{} is expected to cut $\frac{1}{2}m$ edges, so the above theorem says that any improvement in the constant factor requires instances with very high girth. We show how this can be exploited classically:
\newtheorem*{theorem:main2}{\Cref{thm:cycles}}
\begin{theorem:main2}For a connected graph with girth $g$ linear in $m$, the quantity $m-n+1$ is constant.
\end{theorem:main2}
\maxcut{} can be solved exactly in classical polynomial time in this regime through a variety of approaches, and even the simple algorithm of cutting every edge in a spanning tree is guaranteed to cut $m-O(1)$ edges, since $m-n+1$ is precisely the number of edges not in the spanning tree.

\paragraph{Related work.} Patamawisut et al.\ describe a circuit implementation of DQI for \maxcut{}, using a different decoder based on Gauss-Jordan elimination instead of the aforementioned exact decoder based on perfect matching~\cite{PatamawisutEtAl2025CircuitDQI}. While Jordan et al.\ focus on more general cases of Max-LINSAT than \maxcut{}, they do observe a limitation on the performance of DQI for random instances of Max-$2$-XORSAT, which corresponds to $\{-1,+1\}$-weighted \maxcut{}~\cite[Section 13.3]{JordanEtAl2024OptimizationDQI}. In concurrent work with ours, Anschuetz, Gamarnik, and Lu identify obstructions for DQI on random instances of the more general Max-$k$-XORSAT problem, suggesting that a quantum advantage through DQI is not possible on typical random instances~\cite{AnschuetzEtAl2025DQIStructure}. Another line of concurrent work, by Marwaha et al., gives evidence suggesting that the probability distribution from which DQI samples solutions may be hard to simulate classically~\cite{MarwahaEtAl2025DQIComplexity}. However, even if DQI is doing something ``inherently quantum'' that is not efficiently simulable classically, this does not directly speak to whether efficient classical algorithms can outperform DQI. In fact \maxcut{} may be an example where DQI produces a distribution that is not efficiently producible classically, but as we show, independent polynomial-time classical algorithms that outperform it do exist. 

\paragraph{Organization.} The next section motivates and explains DQI for \maxcut{} without assuming any prior knowledge of DQI. \Cref{sec:classical} establishes \Cref{thm:cycles} and demonstrates that \maxcut{} can be solved in classical polynomial time when the girth is linear. We close with a discussion in \Cref{sec:discussion} on whether extensions of DQI might be able to provide approximation advantages for \maxcut{}.

\section{Specializing Decoded Quantum Interferometry for \maxcut{}}
It will be convenient to formulate \maxcut{} as a Boolean polynomial optimization problem. We will think of a cut as partitioning an unweighted graph, $G=(V,E)$ into two pieces by assigning either a $+1$ or $-1$ to each vertex, so that points $z \in \{-1,1\}^V$ of the Boolean hypercube over $V$ correspond to cuts. We take $n \defeq |V|$, $m \defeq |E|$, and $g$ as the girth of $G$, which is the length of a shortest cycle. We implicitly assume that input instances $G$ come from an infinite family $\mathcal{G}$ so that asymptotic quantities are well defined; in particular we assume that $\lim_{m\to\infty} g/m$ exists. The number of edges cut by a $z$, called the \emph{cut value} of $z$, is then 
\begin{equation*}
    h(z) = \frac{1}{2} \sum_{ij \in E} 1 - z_i z_j.
\end{equation*}
\maxcut{} is equivalent to the Boolean polynomial optimization problem
\begin{equation*}
    \max_{z \in \{-1,1\}^V} h(z).
\end{equation*}

For the incidence vector $\gamma \in \{0,1\}^U$ of a set $S \subseteq U$, we take 
\begin{align*}
x^{\gamma} \defeq \prod_{i \in S} x_i^{\gamma_i}. 
\end{align*}
We can represent a polynomial over the Boolean hybercube,
\begin{equation*}
    p(z) = \sum_{\alpha \in \{0,1\}^V} p_{\alpha} z^{\alpha}
\end{equation*}
as a quantum state in two different ways: 
\begin{itemize}
\item Encode the values of $p$ over inputs as amplitudes
\begin{equation*}
    \ket{p} = \unit\left(\sum_{z \in \{-1,1\}^V} p(z) \ket{z}\right),
\end{equation*}
where $\ket{z} \defeq \ket{\frac{1-z_1}{2} \frac{1-z_2}{2} \cdots}$ and $\unit(\cdot)$ outputs a normalized state, or
\item Encode the coefficients of $p$ over monomials as amplitudes
\begin{equation*}
    \ket{\hat{p}} = \unit\left(\sum_{\alpha \in \{0,1\}^V} p_{\alpha} \ket{\alpha}\right).
\end{equation*}
\end{itemize}
These two encodings are related by the Boolean Fourier transform, also known as the Hadamard transform (see \cite{Odonnell2021AnalysisBooleanFunctions} for Fourier analysis of Boolean functions). Hadamard gates applied on all qubits enact this transform on a quantum computer, and we have $H^{\otimes n}\ket{p} = \ket{\hat{p}}$ and $H^{\otimes n}\ket{\hat{p}} = \ket{p}$.

\paragraph{Quantum Fourier Sampling.} The above observation gives a simple quantum algorithm, called Quantum Fourier Sampling (QFS)~\cite{FeffermanUmans2016QFS}, for trying to find a cut $z$ with relatively large value $h(z)$:
\begin{algorithm}[Quantum Fourier Sampling]\ \\
\label{alg:QFS}
\begin{enumerate}
    \item \label{item:prepare} Prepare $\ket{\hat{h}}$
    \item Produce $\ket{h} = H^{\otimes n}\ket{\hat{h}}$
    \item Measure $\ket{h}$ in the computational basis
\end{enumerate}
\end{algorithm}
This algorithm allows sampling a cut $z$ with probability proportional to $h^2(z)$, so that the cut value dictates the probability of picking a cut, as we would want. We will discuss Step \ref{item:prepare} in more detail below, and for now we note that it can be executed efficiently since $h$ is an explicit degree-2 polynomial \cite[Lemma 1]{FeffermanUmans2016QFS}.

What kind of approximation guarantee does this give for \maxcut{}? (Spoiler alert: it's not so great). The expected number of edges cut is $\frac{1}{2}m + O(1)$. This is just a slight improvement over the expectation of $\frac{1}{2}m$ edges achieved by the random assignment where each vertex is independently assigned to a side of the cut with probability $\frac{1}{2}$. 

\subsection{Decoded Quantum Interferometry}
What if we could sample cuts from a distribution proportional to some higher power or maybe even a higher degree polynomial of the objective $h$? This could amplify the probabilities of good cuts, and we might expect a better approximation guarantee. 

This is where DQI steps in. Sampling from higher degree distributions is indeed possible, but the tradeoff is an increase in complexity of preparing the corresponding state in Step \ref{item:prepare} of QFS. This leaves the questions: (i) how do we find and analyze a good polynomial of $h$, and (ii) can we prepare the corresponding state efficiently? DQI answers both of these questions.

For (i), the observation in~\cite{JordanEtAl2024OptimizationDQI} is that if we let $y_{ij} \defeq z_i z_j$, then $h$ is a symmetric polynomial in the variables $y$ since it only depends on $g(y) \defeq \sum_{ij \in E} y_{ij}$. The polynomials 
\begin{align*}
    g_k(y) \defeq \sum_{\substack{\beta \in \{0,1\}^E\\|\beta|=k}} y^\beta
\end{align*}
are a basis for the symmetric polynomials. Hence, any degree-$l$ polynomial of $h$, which is also a degree-$l$ polynomial in the variables $y$, can be expressed as 
\begin{align} \label{eq:q}
q(y) = \sum_{k=0}^l \mu_k g_k(y) = \sum_{\beta \in \mathcal{E}_l} q_\beta y^\beta = \sum_{\alpha \in \{0,1\}^V} q'_\alpha z^\alpha,  
\end{align}
where $\mathcal{E}_k \defeq \{\beta \in \{0,1\}^E \mid |\beta| \leq k\}$. The idea is then to find a good polynomial $q$ and run \Cref{alg:QFS} on it to sample a cut $z$ with probability proportional to $q^2(z)$.
The coefficients $\mu_k$ of a $q$ that maximize the expected cut value, $\mathbb{E}_{z \sim q^2}[h(z)]$, can be found classically in polynomial time by solving a maximum-eigenvector problem for an $(l+1) \times (l+1)$ tridiagonal matrix (see \Cref{lem:DQI-expectation}). In light of this, we can think of the input to DQI as the unweighted graph $G$ and the degree parameter $l$, where the approximation guarantee of DQI monotonically increases with $l$.

How large does $l$ need to be? A distribution supported only on optimal cuts is proportional to 
\begin{align*}
    h^*(y) = \prod_{i = 2}^s (h(y)-\lambda_i),
\end{align*}
where $\lambda_1 > \lambda_2 > \cdots > \lambda_s$ are the possible cut values $h(y)$ can attain. The symmetric polynomial $h^*$, which is proportional to $(h^*)^2$, has degree at most $m$ since the cut value $h(y) \in [0,m]$ for all $y$. Hence, taking $l=m$ would solve \maxcut{}, and it suffices to choose $l \leq m$. This also suggests that DQI should not run in polynomial time when $l=m$, since \maxcut{} is \NP-complete, and it is widely believed that polynomial-time quantum algorithms cannot solve \NP-complete problems.  DQI gives a better expected approximation guarantee than the best classical algorithm currently known for certain problems in the regime where $l = \Theta(n) = \Theta(m)$ (for more general constraint satisfaction problems, $n$ and $m$ are the number of variables and constraints). 

\paragraph{Preparing the initial state for QFS.}
We need to be able to prepare $\ket{\hat{q}}$ in polynomial time. We can do so readily in a Hilbert space where we have $m$ qubits corresponding to $E$. First note that $g_k$ can be represented using a Dicke state:
\begin{align*}
    \ket{\hat{g_k}}_E = \frac{1}{\sqrt{\binom{m}{k}}}\sum_{\substack{\beta \in \{0,1\}^E:\\|\beta|=k}} \ket{\beta},
\end{align*}
where the subscript on the ket identifies the set that the qubits correspond to. Then we obtain, as a superposition of Dicke states
\begin{align*}
    \ket{\hat{q}}_E = \unit\left(\sum_{k=0}^l \mu_k \sqrt{\binom{m}{k}} \ket{\hat{g_k}}\right) = \unit\left(\sum_{\beta \in \mathcal{E}_l}q_\beta \ket{\beta}\right),
\end{align*}
which can be prepared efficiently as described in~\cite{JordanEtAl2024OptimizationDQI}. However, in order to apply Quantum Fourier Sampling, we instead need the state
\begin{align*}
    \ket{\hat{q}}_V = \unit\left(\sum_{\alpha \in \{0,1\}^V}q'_\alpha \ket{\alpha}\right).
\end{align*}

How can we construct $\ket{\hat{q}}_V$ from $\ket{\hat{q}}_E$?  Since $y_{ij} = z_i z_j$, we get a natural map $f$ from $\beta \in \{0,1\}^E$ to $\alpha \in \{0,1\}^V$ as 
\begin{align*}
y^\beta = \prod_{\substack{ij \in E:\\\beta_{ij} = 1}} z_i z_j = z^{f(\beta)},
\end{align*}
where $z_i^2 = 1$ for all $i$ is employed to reduce the monomial $z^{f(\beta)}$ so that $f(\beta) \in \{0,1\}^V$. We think of $\beta$ as an incidence vector representing a subgraph (i.e., a subset of edges) and the \emph{parity vector}, $f(\beta)$ as indicating the parity of the degree of vertex in $V$ in the subgraph $\beta$ (vertices may have no incident edges). We can see that $f$ is not injective in general by letting $G$ be a cycle on three vertices, for which $f(1,1,1) = f(0,0,0) = (0,0,0)$. In other words, all vertices in $V$ have even degree in both $G$ itself and the empty subgraph. 

\begin{lemma}[specialized for \maxcut{} from \cite{JordanEtAl2024OptimizationDQI}]
\label{lem:uncompute}
Suppose that $G$ is a graph where each parity vector $\alpha \in \{0,1\}^V$ arises from at most one subgraph $\beta \in \mathcal{E}_l$. Then $f$ is injective (on the domain $\mathcal{E}_l$), and the following algorithm may be used to produce $\ket{\hat{q}}_V$.
\begin{enumerate}
    \item Prepare $\ket{\hat{q}}_E\ket{0}_V = \unit\left(\sum_{\beta \in \mathcal{E}_l}q_\beta \ket{\beta}_E\ket{0}_V\right)$
    \item Apply $f$ coherently to get: $\unit\left(\sum_{\beta \in \mathcal{E}_l}q_\beta \ket{\beta}_E\ket{f(\beta)}_V\right)$
    \item \label{item:uncompute} Uncompute the $E$ register using $f^{-1}$ to yield: $\ket{\psi} \defeq \unit\left(\sum_{\beta \in \mathcal{E}_l}q_\beta \ket{0}_E\ket{f(\beta)}_V\right)$
\end{enumerate}
\end{lemma}
\begin{proof} We see from \Cref{eq:q} that 
\begin{align} \label{eq:q'}
q'_\alpha = \sum_{\substack{\beta \in \mathcal{E}_l: \\ f(\beta) = \alpha}} q_\beta,
\end{align}
with $q'_\alpha = 0$ if there is no $\beta \in \mathcal{E}_l$ such that $f(\beta) = \alpha$. By \Cref{eq:q'}, discarding the $E$ register from $\ket{\psi}$ produces the desired state:
\begin{align*}
\unit\left(\sum_{\beta \in \mathcal{E}_l}q_\beta \ket{f(\beta)}_V\right) =
\unit\left(\sum_{\alpha \in \{0,1\}^V}q'_\alpha \ket{\alpha}\right).
\end{align*}
\end{proof}

DQI uses the algorithm of \Cref{lem:uncompute} and thus requires $f$ to be injective. What kind of structural assumption does this impose on the input $G,l$?

\begin{lemma} \label{lem:girth}
    Suppose we have a graph $G$ and $l \in \mathbb{N}$. Then $f$ is injective on $\mathcal{E}_l$ if and only if $G$ has girth at least $2l+1$.
\end{lemma}
\begin{proof}
    First we observe that $f$ is linear when viewed as a function on $\mathbb{F}_2^E$. Let $B \in \mathbb{F}_2^{E \times V}$ be the edge-vertex incidence matrix of $G$, with $B_{e,v} = 1$ if and only if the edge $e$ is incident to the vertex $v$. Then $f(\beta) = \alpha$ if and only if $\alpha = B^T \beta,$ with $\alpha$ and $\beta$ taken as vectors over $\mathbb{F}_2$.
    
    We show that if $f$ is not injective then the girth of $G$ is at most $2l$.  Suppose there are two distinct subgraphs $\beta, \beta' \in \mathcal{E}_l$ with $f(\beta) = f(\beta')$. Then, in $\mathbb{F}_2^E$, $\gamma \defeq \beta + \beta' \not= \mathbf{0}.$ The parity vector of $\gamma$ is $\mathbf{0}$, since $f(\gamma) = f(\beta) + f(\beta') = \mathbf{0}$. This means that the nonempty subgraph corresponding to $\gamma$ has even degree at all vertices and has at most $2l$ edges. This subgraph must therefore contain a cycle of size at most $2l$. To see this, one can start at a vertex and keep moving to other vertices using previously unused edges. Since all degrees are even, it will always be possible to leave a vertex that is visited. Since the graph is finite, one must eventually revisit a vertex, resulting in a cycle.

    Now we show the opposite direction. Suppose there is a cycle with at most $2l$ edges. Break the cycle into two edge-disjoint paths of length at most $l$, both between the same pair of vertices $i$ and $j$ on the cycle. Represent these paths by incidence vectors $\beta,\beta' \in \mathcal{E}_l$. Then $f$ is not injective since $f(\beta) = f(\beta') = \alpha$, where $\alpha_i = \alpha_j = 1$ and all other entries are 0.  
\end{proof}

So for \maxcut{} there is a simple graph-theoretic characterization of the structure DQI requires. It makes sense to pick the largest $l$ such that $f$ remains injective; however, we also need an efficient classical algorithm to implement $f^{-1}$, which may limit the choice of $l$. For \maxcut{}, it turns out that $f^{-1}$ can be computed in classical polynomial time for any value of $l$, as we show in the next section. In this case $l$ can just be taken to be the largest integer so that $G$ has girth at least $2l+1$. We are now in position to present the full DQI algorithm for \maxcut{}.

\begin{algorithm}[Decoded Quantum Interferometry for \maxcut{}]\ \\
\label{alg:DQI}
\emph{Input:} unweighted graph $G = (V,E)$
\begin{enumerate}
    \item \label{item:compute-girth} 
    Compute the girth $g$ of $G$, and set $l \defeq \lfloor \frac{g-1}{2} \rfloor$
    \item \label{item:q:eigenvector} 
    Solve an $(l+1)\times(l+1)$ maximum eigenvector problem to find the coefficients of an optimal degree-$l$ symmetric polynomial $q$ in $y$
    \item \label{item:uncompute-step}
    Run the algorithm of \Cref{lem:uncompute} to prepare $\ket{\hat{q}}_V$
    \item \label{item:QFS}
    Run QFS (\Cref{alg:QFS}) on $\ket{\hat{q}}_V$ to sample a cut $z$ with probability proportional to $q^2(z)$
\end{enumerate}
\end{algorithm}

\subsection{Solving the classical portion of DQI}
Solving the following problem suffices to implement the uncompute step (Step \ref{item:uncompute}) of the algorithm of \Cref{lem:uncompute}.
\begin{problem}[DQI uncomputation for \maxcut{}]\ \\
\label{prob:decode}
\emph{Input:} A parameter $l \in \mathbb{N}$, an unweighted graph $G=(V,E)$, and a vector $\alpha \in \mathbb{F}_2^V$\\
\emph{Output:} A subgraph with at most $l$ edges, represented as $\beta \in \mathbb{F}_2^E$, whose degree parities match $\alpha$ (i.e., $|\beta| \leq l$ and $f(\beta) \defeq B^T \beta = \alpha$, where $B \in \mathbb{F}_2^{E \times V}$ is the edge-vertex incidence matrix of $G$)
\emph{Remark:} Additionally assuming that $G$ has girth at least $2l+1$ guarantees a unique solution for feasible instances by \Cref{lem:girth}.
\end{problem}

\begin{remark}[Classical decoding in DQI]
\label{rem:DQI-decoding}
DQI is originally framed in terms of a classical decoding problem for an error-correcting code. For \maxcut{} this perspective stems from the edge-incidence matrix $B$. The column space of $B$ generates the \emph{cut space} of $G$, which is the span over $\mathbb{F}_2$, of incidence vectors of cuts in $G$. Since any cycle crosses a cut an even number of times, the cut space is orthogonal to the \emph{cycle space}, generated in the same fashion by incidence vectors of cycles. This space may be viewed as a cycle code space and is characterized as incidence vectors of subgraphs of $G$ which have even degree at every vertex. The parity-check matrix of this code is $B^T$, and the girth of $G$ is the distance of the code. This code is usually called a graph code or cycle code. 

In this picture, \Cref{prob:decode} is about finding a short vector generating a given syndrome, and as observed in \cite{JordanEtAl2024OptimizationDQI} it can also be solved as a more standard maximum-likelihood decoding problem. Finally, observe the following duality: for DQI to solve an optimization problem over cuts (\maxcut{}), it must solve a kind of decoding problem over the dual cycle space (\Cref{prob:decode}). See, e.g., Section 1.9 of \cite{Diestel2025GraphTheory} for more on the cut and cycle spaces of a graph.
\end{remark}

We illustrate how \Cref{prob:decode} can be cast as a $T$-join problem~\cite{EdmondsJohnson1973ChinesePostman}, which can be solved in polynomial time by a reduction to a perfect matching problem (see, e.g., \cite[Theorem 29.1]{Schrijver2003CombinatorialOptimization}).
\begin{problem}[Minimum $T$-join]\ \\
\label{prob:T-join}
\emph{Input:} An unweighted graph $G=(V,E)$ and a subset of vertices $T \subseteq V$\\
\emph{Output:} A subgraph $H$ of $G$, with the minimum number of edges, such that set of odd-degree nodes in $H$ is precisely $T$
\end{problem}

\begin{fact}\label{lem:poly-time-solvability}
\Cref{prob:decode} is polynomial-time reducible to \Cref{prob:T-join}. 
\end{fact}
\begin{proof} Given an instance of \Cref{prob:decode}, set $T \defeq \{i \in V \mid \alpha_i = 1\}$ and solve \Cref{prob:T-join} on $G,T$. If this instance of \Cref{prob:T-join} is infeasible, or if $H$ has more than $l$ edges, then \Cref{prob:decode} is infeasible. Otherwise, take $\beta$ as the incidence vector of $H$.
\end{proof}

This is the last piece in establishing:
\begin{theorem}DQI for \maxcut{} (\Cref{alg:DQI}) is correct and runs in polynomial time on any input graph $G$.
\end{theorem}
\begin{proof}
The girth of an unweighted undirected graph can be found by running $n$ breadth-first searches, one rooted at each vertex~\cite{ItaiRodeh1978MinCycle}, so Steps \ref{item:compute-girth} and \ref{item:q:eigenvector} run in classical polynomial time. Since $g \geq 2l+1$, $f$ is injective by \Cref{lem:girth}, and Step~\ref{item:uncompute-step} correctly produces $\ket{\hat{q}}_V$. This step also runs in polynomial time since \Cref{prob:decode} runs in classical polynomial time by \Cref{lem:poly-time-solvability}. QFS (Step~\ref{item:QFS}) runs in polynomial time and samples from $q^2$ as explained in~\cite{FeffermanUmans2016QFS}. The fact that Step~\ref{item:q:eigenvector} produces an optimal symmetric degree-$l$ polynomial $q$ comes from the proof of~\cite[Lemma 9.2]{JordanEtAl2024OptimizationDQI}, although this is not necessary to prove the approximation guarantee of \Cref{alg:DQI}.
\end{proof}
We now move on to considering the approximation ratio that \Cref{alg:DQI} guarantees.

\subsection{Approximation guarantee}
The expected number of edges cut by DQI beyond the random-assignment guarantee of $\frac{1}{2}m$ is captured by the maximum eigenvalue of a tridiagonal matrix: 

\begin{lemma}[Specialization of Lemma 9.2 in \cite{JordanEtAl2024OptimizationDQI} for \maxcut{}] 
\label{lem:DQI-expectation}
Let $q$ be the polynomial obtained in Step \ref{item:q:eigenvector} of \Cref{alg:DQI}. Then the expected number of edges cut by \Cref{alg:DQI} is
\begin{align} \label{eq:DQI-approx}
    \mathbb{E}_{z \sim q^2}[h(z)] = \frac{1}{2}(m + \lambda_{\max}(A^{(m,l)})), 
\end{align}    
    where
\begin{align*}
A^{(m,l)} \;\defeq\;
\begin{bmatrix}
 & a_1 &              &        &        \\
 a_1 &               & a_2 &        &        \\
                 & a_2 &              & \ddots &        \\
                 &               & \ddots       &        & a_l \\
                 &               &              & a_l &
\end{bmatrix}, 
\end{align*}
with $a_k = \sqrt{k(m-k+1)}$ for $1 \leq k \leq l$.
\end{lemma}
We show that DQI needs instances with $l = \Omega(m)$ in order to attain a asymptotically nontrivial approximation guarantee.
\begin{lemma}
    For matrices $A^{(m,l)}$ with $l = o(m)$, $\lambda_{\max}(A^{(m,l)}) = o(m)$.
\end{lemma}
\begin{proof}
Since $A^{(m,l)}$ is nonnegative, $\lambda_{\max}(A^{(m,l)}) \leq \sqrt{m}\lambda_{\max}(B^{(l)})$, where
\begin{align*}
B^{(l)} \;\defeq\;
\begin{bmatrix}
 & \sqrt{1} &              &        &        \\
 \sqrt{1} &               & \sqrt{2} &        &        \\
                 & \sqrt{2} &              & \ddots &        \\
                 &               & \ddots       &        & \sqrt{l} \\
                 &               &              & \sqrt{l} &
\end{bmatrix}.
\end{align*}
The matrix, $\frac{1}{\sqrt{2}}B^{(l)}$ is the $(l+1)\times(l+1)$ Jacobi matrix corresponding to the three-term recurrence relation for the normalized monic Hermite polynomials~\cite[3.5(vi)]{NISTDLMF2025}. The maximum eigenvalue of $\frac{1}{\sqrt{2}}B^{(l)}$ is consequently the largest zero of the degree-$(l+1)$ monic Hermite polynomial, which is $O(\sqrt{l})$~\cite[18.16(v)]{NISTDLMF2025}. Thus $\lambda_{\max}(A^{(m,l)}) = O(\sqrt{ml})$. Since $l = o(m)$, we have $O(\sqrt{ml}) = o(m)$.
\end{proof}

From the two lemmas above, we have:
\begin{theorem}\label{thm:DQI-needs-high-girth} If the expected number of edges cut by DQI in a graph $G \in \mathcal{G}$ is $(\frac{1}{2}+\varepsilon)m$, for some constant $\varepsilon > 0$, the girth $g$ must be linear in $m$.
\end{theorem}
\begin{proof}If we do not have linear girth, $g = \Omega(m)$, then $g = o(m)$ since $\lim_{m \to \infty} g/m$ exists. In this case $l \leq \frac{g-1}{2} = o(m)$, and by the above lemmas, the expected number of edges cut by DQI is $\frac{1}{2}(1 + o(1))m$.
\end{proof}
\begin{remark}
The approximation guarantee for DQI is expressed with respect to the trivial upper bound, $m$ on the optimal cut since DQI does not leverage more refined upper bounds such as those arising from convex programs. However, this means that DQI provides a lower bound on the number of edges in a maximum cut and more generally on the number of satisfiable constraints in a constraint satisfaction problem, even when DQI is not guaranteed to run in polynomial time.   
\end{remark}

\section{Classical solvability of high-girth instances}
\label{sec:classical}
High girth has been previously exploited for solving \maxcut{} both classically and using the QAOA \cite{ElAlaouiEtAl2023LocalAlgorithms,ThompsonEtAl2022HighGirthMaxCut,FarhiEtAl2025QAOAHighGirth}; however, ``high'' in such contexts typically refers to constant or perhaps logarithmic girth, as expected in some kinds of random graphs. In contrast, in order for DQI to guarantee $(\frac{1}{2}+\varepsilon)m$ edges in expectation, we know that we must have girth $g = \Theta(n) = \Theta(m)$. We show that in this regime, \maxcut{} can be solved exactly in polynomial time. In fact, the approaches we describe can also solve weighted instances, while it is nontrivial to extend DQI to handle weighted instances.

\begin{lemma} \label{lem:tree-decomp}
Suppose $G \in \mathcal{G}$ is a connected graph with girth $g = \Omega(m)$. Then the vertices of $G$ can be partitioned into vertex-disjoint trees $T_1,\ldots,T_t$ where: (i) $t = O(1)$, and (ii) there is at most one edge in $G$ between any two $T_i$ and $T_j$.
\end{lemma}
\begin{proof}
Set $d \defeq \lfloor \frac{g-3}{4} \rfloor$, and find a spanning tree $T$ of $G$. Then run the following algorithm.
\smallbreak
\noindent While $T$ has depth at least $d$:
\begin{enumerate}
\item Find the deepest leaf $u$ in $T$, and let $S_i$ be the subtree rooted at the ancestor of $u$ at distance $d$ from $u$
\item Let $T_i$ be the subgraph of $G$ induced by the vertices of $S_i$
\item Remove $S_i$ from $T$ and set $i = i+1$
\end{enumerate}
When the loop has terminated, set $T_i$ to be subgraph induced by $T$ itself. Let $t \defeq i$ at this point.

We are guaranteed that each tree $S_i$, except possibly $S_t$, has depth exactly $d$, since we picked $u$ to be the deepest leaf in $T$, and therefore in $S_i$, in each iteration. The last tree $S_t$ has depth at most $d$. 

To prove (i), since the trees $S_1,\ldots,S_{t-1}$ have depth $d$, these disjoint trees each have at least $d$ edges. This gives us $t = O(1)$ since $g = \Omega(m)$ implies $d = \Omega(m)$ (the hidden constants for $t$ and $d$ only depend on those for $g$).

Next observe that each $T_i$ can contain no edge of $G-T$ and is consequently a tree. Since each $T_i$ is connected by edges in $T$ and has depth at most $d$, if one contained some edge in $G-T$ there would be a cycle with at most $2d+1 < g$ edges, contradicting that the girth is $g$. The proof of (ii) is similar to this. Two edges between distinct $T_i$ and $T_j$ would create a cycle of length at most $4d+2 < g$.
\end{proof}

\begin{theorem}\label{thm:cycles}
For a connected graph $G \in \mathcal{G}$ with girth $g = \Omega(m)$, the quantity $\mu \defeq m-n+1$ is constant.
\end{theorem}
\begin{proof}
Let $T$ be a spanning tree, and 
construct $T_1,\ldots,T_t$ from $T$ as specified in \Cref{lem:tree-decomp}. Since each $T_i$ is a subgraph of $T$, each edge in $G-T$ must go between distinct $T_i$ and $T_j$. Conditions (i) and (ii) then imply that there is only a constant number of such edges. In other words, $\mu$ is constant.
\end{proof}

\begin{corollary}
    \maxcut{} can be solved in polynomial time on a graph $G \in \mathcal{G}$ with girth $g = \Omega(m)$.
\end{corollary}
\begin{proof}
Each component of $G$ must also have girth $\Omega(m)$, so we can apply \Cref{thm:cycles} to each component individually. \maxcut{} can be solved in time $O(m+n)\min\{2^{m/5},2^{(m-n)/2}\}$ on a connected graph~\cite[Theorem 2]{ScottSorkin2003FasterMaxCut}.
\end{proof}

\paragraph{Other approachs for high-girth \maxcut{}.} \Cref{thm:cycles} enables other approaches for solving \maxcut{} in polynomial time. The fundamental cycles of a spanning tree $T$ are the set of $\mu$ unique cycles created by adding each of the $\mu$ non-tree edges to $T$. These cycles form a basis for all cycles over $\mathbb{F}_2^E$ (see \Cref{rem:DQI-decoding}). Since $\mu$ is constant for graphs with linear girth, so are the total number of cycles, which is at most $2^\mu$. This allows solving positive-weighted \maxcut{} as a binary integer program with a constant number of constraints:
\begin{align*}
    \max \qquad &\sum_{ij \in E} w_{ij} x_{ij} \\
    \text{subject to} \quad &\sum_{ij \in C} x_{ij} \leq |C|-1, \; \forall\text{odd cycles }C \\
    &x_{ij} \in \{0,1\},\; \forall ij \in E,
\end{align*}
which can be done in polynomial time~\cite{Papadimitriou1981IntegerProgramming}. The constraints ensure that a bipartite subgraph is selected by preventing picking all the edges of any odd cycle.

Perhaps the simplest way to solve high-girth instances of \maxcut{} is via a dynamic program. Pick a spanning tree $T$. Cutting all the edges of the tree already gives a cut of size $m - O(1)$, since $\mu$ is constant. Enumerate all ways of assigning $\{-1,+1\}$ to the $\mu$ non-tree edges. For each assignment, one can solve a dynamic program that works bottom up from the leaves to find the best cut of the edges of $T$ assuming that some nodes have been fixed by the assignment. This works for weighted \maxcut{} as well.

\section{Discussion} 
\label{sec:discussion}
A natural question is whether some extension of DQI might be able to get around the barriers we presented. One can consider probabilistic decoding approaches for DQI that need not always succeed~\cite{JordanEtAl2024OptimizationDQI,ChaillouxTillich2025SoftDecoders}. Such approaches can conceivably tolerate some cycles of smaller girth, possibly extending the kinds of instances of \maxcut{} that DQI can address. However, we observe that we only need $\mu = O(1)$ for classical solvability, and some of the classical approaches from \Cref{sec:classical} work even when $\mu = \Theta(\log n)$. This, for instance, implies that \maxcut{} on any graph with a logarithmic number of cycles (of any size) can be solved efficiently. Some instances with a polynomial number of cycles may also be solvable, since it is possible to have a polynomial cycles with a fundamental cycle basis of size $\mu = \Theta(\log n)$. Thus classical approaches seem fairly robust in terms of handling cycles, and it is not clear if DQI would be able to do better.

Is \maxcut{} an isolated case, or are there more general ways for classical algorithms to leverage the same structure that DQI does? While the particular property we leverage for \maxcut{} does feel very special, the question is currently tough to answer. The prime candidates for a DQI advantage, such as Optimal Polynomial Intersection, should be further vetted by attempts at designing better classical approximation algorithms in the regimes where running DQI is feasible and efficient.

\section*{Acknowledgements}
I thank Stephen Jordan and Noah Shutty for helpful discussions on DQI and feedback on this work.

This work is supported by the U.S.\ Department of Energy (DOE), Office of Science, National Quantum Information Science Research Centers, Quantum Systems Accelerator (QSA). Additional support is acknowledged from DOE, Office of Science, Accelerated Research in Quantum Computing, Fundamental Algorithmic Research toward Quantum Utility (FAR-Qu). 

This article has been authored by an employee of National Technology \& Engineering Solutions of Sandia, LLC under Contract No.\ DE-NA0003525 with the U.S. Department of Energy (DOE). The employee owns all right, title and interest in and to the article and is solely responsible for its contents. The United States Government retains and the publisher, by accepting the article for publication, acknowledges that the United States Government retains a non-exclusive, paid-up, irrevocable, world-wide license to publish or reproduce the published form of this article or allow others to do so, for United States Government purposes. The DOE will provide public access to these results of federally sponsored research in accordance with the DOE Public Access Plan \url{https://www.energy.gov/downloads/doe-public-access-plan}.

\printbibliography
\end{document}